\documentclass[preprint,1p,times,onecolumn]{elsarticle}

\usepackage{graphicx}
\usepackage{amssymb}
\usepackage{amsmath}
\usepackage{amsthm}
\usepackage{amsfonts}
\usepackage[hyphens]{url}
\usepackage[hidelinks]{hyperref}
\usepackage{tikz}
\usetikzlibrary{positioning}
\usepackage{booktabs}
\usepackage{clrscode3e}
\usepackage{enumerate}
\usepackage[shortlabels]{enumitem}
\usepackage{subfigure}

\DeclareMathOperator{\Cartesian}{\Box} 
\DeclareMathOperator{\Direct}{\times} 
\DeclareMathOperator{\Strong}{\boxtimes} 
\DeclareMathOperator{\Lexicographic}{\circ} 
\DeclareMathOperator{\Kron}{\otimes} 
\DeclareMathOperator{\Isomorph}{\simeq} 
\DeclareMathOperator{\mreduce}{\leq_\mathrm{M}} 
\DeclareMathOperator{\treduce}{\leq_\mathrm{T}} 
\newcommand{\mat}[1]{\ensuremath{\mathbf{#1}}} 
\newcommand{\tr}[1]{\ensuremath{#1}^\intercal} 
\newcommand{\adj}[1]{\ensuremath{\mathbf{Adj(\mathnormal{#1})}}} 
\newcommand{\gi}[1]{{\bf{GI[#1]}}} 
\newcommand{\primality}[1]{{\bf{Primality[#1]}}} 
\newcommand{\compositeness}[1]{{\bf{Compositeness[#1]}}} 

\newcommand{\ansyes}{\textsc{yes}} 
\newcommand{\ansno}{\textsc{no}} 
\newcommand{\GG}{\Theta} 
\newcommand{\conn}{C} 
\newcommand{\unc}{U} 
\newcommand{\nonbip}{NB} 

\newtheorem{teo}{Theorem}[section]
\newtheorem{lemma}[teo]{Lemma}

\newtheorem{corol}[teo]{Corollary}
\newtheorem{defi}[teo]{Definition}

\newtheorem{obs}[teo]{Observation}

\newcommand{\Z}{\mathbb{Z}}
\newcommand{\N}{\mathbb{N}}

\journal{arXiv}

\begin{document}

\begin{frontmatter}

\title{Direct Product Primality Testing of Graphs is GI-hard}

\author[disi]{Luca~Calderoni\corref{cor1}}
\ead{luca.calderoni@unibo.it}

\author[disi]{Luciano~Margara}
\ead{luciano.margara@unibo.it}

\author[disi]{Moreno~Marzolla}
\ead{moreno.marzolla@unibo.it}

\cortext[cor1]{Corresponding author}
\address[disi]{Department of Computer Science and Engineering, University of Bologna, Italy}

\begin{abstract}
We investigate the computational complexity of the graph primality testing problem with respect to the direct product (also known as Kronecker, cardinal or tensor product). In~\cite{imrich1998} Imrich proves that both primality testing and a unique prime factorization can be determined in polynomial time for (finite) connected and nonbipartite graphs. The author states as an open problem how results on the direct product of nonbipartite, connected graphs extend to bipartite connected graphs and to disconnected ones. In this paper we partially answer this question by proving that the graph isomorphism problem is polynomial-time many-one reducible to the graph compositeness testing problem (the complement of the graph primality testing problem). As a consequence of this result, we prove that the graph isomorphism problem is polynomial-time Turing reducible to the primality testing problem. Our results show that connectedness plays a crucial role in determining the computational complexity of the graph primality testing problem.
\end{abstract}


\end{frontmatter}


\section{Introduction}\label{sec:intro}
Factorization is a fundamental task in mathematics and in many other disciplines 
including  computer science, physics and engineering. 
The notion of product among mathematical objects not only enables the
creation of new objects from smaller ones, but also
naturally addresses the more complex task of decomposing an object as
the product of simpler components. 
Factoring a mathematical object is therefore one of the the main methods for deeply understanding
its structure.

Integer factorization is by far
the most widely known and studied factorization problem; however,
many other types of mathematical objects 
have been extensively studied in order to understand if and how they can be factored.
Specifically, graph factorization
with respect to several notions of product has been thoroughly investigated
both from the theoretical and from the practical point of view.

In this paper we investigate the computational complexity of 
graph factorization with respect to the direct product (see Definition~\ref{def:direct-product})
which is one of the most widely studied graph product. 
Some authors refer to the direct product as the Kronecker, tensor or cardinal product. 
We will name it direct product and we will denote it by the operator~$\Direct$ 
according to the notation used 
in the recent book
by Hammack, Imrich and Klav{\v{z}}ar~\cite{hammack2011handbook}.
 
Direct product is one of the three products  (the other two being the Cartesian and Strong products) 
that satisfies the following fundamental algebraic properties 
($\simeq$ stands for "isomorphic to"):
\begin{enumerate}
\item Commutativity: $G_1 \Direct G_2 \Isomorph G_2 \Direct G_1$
\item Associativity: $G_1 \Direct\, (G_2 \Direct G_3) \Isomorph (G_1 \Direct G_2)\Direct G_3$
\item Projections from a product to its factors are weak homomorphisms
\end{enumerate}

A fourth product have been considered in the literature, namely the lexicographic product.
Lexicographic product does not satisfy properties~1 and~3.

We both consider the primality testing problem and the
factorization problem.  Informally, primality testing is a decision problem that, given a graph $G$, answers the question: ``is $G$ the product of smaller, nontrivial graphs?''. Graph factorization aims at decomposing~$G$ into the product of smaller nontrivial graphs (more formal definitions will be given in the next section).
Although factorization of general with respect to the direct product is not unique,
Imrich~\cite{imrich1998} proved that if a graph is connected and nonbipartite, then 
its factorization with respect to the direct product is unique and can be computed 
in polynomial time. In this paper we address the following question posed by Imrich at the end of his
paper.

\begin{quote}
{\em
How do results on the cardinal product of nonbipartite, connected graphs extend to bipartite connected graphs and to disconnected ones ?}
\end{quote}

We prove (Theorem~\ref{th:GIhard})  that the graph isomorphism problem
reduces to the problem of testing the compositeness of possibly unconnected, nonbipartite graphs. 
Since the reduction we use is a polynomial time many-one reduction, we show (Corollary~\ref{corol:turing-reduction}) that 
 testing the primality of a graph is $GI$-hard. In other words, we prove that  testing the primality 
 of a graph in polynomial time would provide a polynomial time algorithm for testing graph isomorphism, which
 is widely considered to be not feasible, although no formal proof exists.
It remains an open question whether testing primality of bipartite, connected 
graphs can be done in polynomial time

This paper is organized as follows. In section~\ref{sec:preliminaries} we
introduce the notation and definition of terms used in the rest of this work.
In section~\ref{sec:related} we review the relevant literature related to
the graph factorization problem. Section~\ref{sec:results} presents
the main result of this paper. Finally, conclusions and future research
directions are discussed in section~\ref{sec:conclusions}.

\section{Notation and Basic Definitions}\label{sec:preliminaries}

In this section we give basic notation and definitions that will be used throughout the paper.
An undirected graph $G = (V, E)$ is described as a finite set~$V$ of
nodes $V = \{v_1, \ldots, v_n\}$  and a finite set of
edges $E \subseteq V \times V$, where an edge $e \in E$ is an
unordered pair of nodes $e = \{u, v\}$, $u, v \in V$. Given a
graph~$G$, $V(G)$ and~$E(G)$ denote the set of nodes and edges of~$G$,
respectively.
We denote by $G_1 \cup G_2$ the disjoint union of graphs~$G_1$ and~$G_2$, i.e., 
the graph with node set $V(G_1) \cup V(G_2)$ and edge set 
$E(G_1) \cup E(G_2)$. Disjoint means that~$V(G_1)$ and~$V(G_2)$ 
satisfy $V(G_1)\cap V(G_2)=\emptyset$.

The set of edges of a graph~$G$ can be represented also as an
adjacency matrix~$\mat{M}$. If~$G$ has~$n$ nodes, $\mat{M}$ is a $n \times n$ binary matrix, such that
$\mat{M}_{ij} = 1$ if and only if  $\{v_i, v_j\}\in E$.
The adjacency matrix for undirected graphs is symmetric,
since every edge $\{v_i, v_j\}$ can also be written as $\{v_j, v_i\}$.
As a shorthand notation, we denote with $\adj{G}$ the
adjacency matrix of graph~$G$.

We use the symbol~$\Gamma$ to denote the set of finite, undirected
graphs where no self-loops are allowed. 
The symbol~$\Gamma_0$ denotes the set of finite,
undirected graphs where self-loops are allowed; a self-loop is
an edge of the form $\{v, v\}$, for some $v \in V(G)$.

Four types of graph products have been investigated in the literature:
\emph{Cartesian product}, \emph{Direct product}, \emph{Strong product}
and \emph{Lexicographic product}. In all cases, the product of two
graphs $G_1, G_2$ is a new graph~$G$ whose set of nodes is the
Cartesian product of~$V(G_1)$ and~$V(G_2)$:
\begin{align*}
V(G) &= V(G_1) \times V(G_2) = \{ \{u, v\}\ |\ u \in V(G_1) \wedge v \in V(G_2) \}
\end{align*}

The edge set~$E(G)$ is defined according to the notion of graph product
as follows.

\begin{defi}[Cartesian product]\label{def:cartesian-product}
  The Cartesian product of two graphs $G_1, G_2$ 
  is denoted
  as $G = G_1 \Cartesian G_2$, where $V(G) = V(G_1) \times V(G_2)$ and
  \begin{align*}
    E(G) &= \left\lbrace \{ \{x, y\}, \{x', y'\} \}\ |\ (x=x' \wedge \{y, y'\} \in E(G_2)) \vee (\{x, x'\} \in E(G_1) \wedge y=y') \right\rbrace
  \end{align*}
\end{defi}

\begin{defi}[Direct product]\label{def:direct-product}
  The direct product of two graphs $G_1, G_2$
  is denoted as
  $G = G_1 \Direct G_2$, where $V(G) = V(G_1) \times V(G_2)$ and
  \begin{align*}
    E(G) &= \left\lbrace \{ \{x, y\}, \{x', y'\} \}\ |\ \{x, x'\} \in E(G_1) \wedge \{y, y'\} \in E(G_2) \right\rbrace
  \end{align*}
  The direct product is also known as Kronecker or cardinal
  product.
\end{defi}

\begin{defi}[Strong product]\label{sub:strong-product}
  The strong product of two graphs $G_1, G_2$
  is denoted
  as $G = G_1 \Strong G_2$, where $V(G) = V(G_1) \times V(G_2)$ and
  \begin{align*}
    E(G) &= E(G_1 \Cartesian G_2) \cup E(G_1 \times G_2)
  \end{align*}
\end{defi}

\begin{defi}[Lexicographic product]\label{def:lexicographic-product}
  The lexicographic product of two graphs $G_1, G_2$
    is denoted as
  $G = G_1 \Lexicographic G_2$, where $V(G) = V(G_1) \times V(G_2)$ and
  \begin{align*}
    E(G) &= \left\lbrace \{ \{x, y\}, \{x', y'\} \}\ |\ \{x, x'\} \in E(G_1) \vee ( x = x' \wedge \{y, y'\} \in E(G_2) ) \right\rbrace
  \end{align*}
\end{defi}

Figure~\ref{fig:product-example} shows the Cartesian, direct, strong and lexicographic product of two graphs $G_1, G_2$. 

\begin{figure}[t]
  \centering

\subfigure[Cartesian ($G_1 \Cartesian G_2$)]{%
  \begin{tikzpicture}[scale=.8]
    \draw (1,1) grid (3,2); 
    \draw (1,0) grid (3,0); \node at (1.5,-0.5) {$G_2$};
    \draw (0,1) grid (0,2); \node at (-0.5,1.5) {$G_1$};    
    \draw (1,2) edge [out=10,in=80,distance=1cm] (1,2);
    \draw (2,2) edge [out=10,in=80,distance=1cm] (2,2);
    \draw (3,2) edge [out=10,in=80,distance=1cm] (3,2);
    \draw (2,1) edge [out=-10,in=-80,distance=1cm] (2,1);
    \foreach \posx in {1, 2, 3}
    \foreach \posy in {1, 2}
    \draw[fill=white] (\posx, \posy) circle (1mm);;
    \draw (0,2) edge [out=10,in=80,distance=1cm] (0,2);
    \foreach \posy in {1, 2}
    \draw[fill=white] (0, \posy) circle (1mm);
    \draw (2,0) edge [out=-10,in=-80,distance=1cm] (2,0);
    \foreach \posx in {1, 2, 3}
    \draw[fill=white] (\posx, 0) circle (1mm);
\end{tikzpicture}}%
\qquad  
\subfigure[Direct ($G_1 \Direct G_2$)]{%
  \begin{tikzpicture}[scale=.8]
    \foreach \posx in {1, 2} {
      \foreach \posy in {1} {
        \draw (\posx, \posy) -- (\posx+1,\posy+1);
        \draw (\posx, \posy+1) -- (\posx+1,\posy);
      }
    }
    \draw (2,2) -- (1,2);
    \draw (2,2) -- (2,1);
    \draw (2,2) -- (3,2);
    \draw (1,0) grid (3,0); \node at (1.5,-0.5) {$G_2$};
    \draw (0,1) grid (0,2); \node at (-0.5,1.5) {$G_1$};    
    \draw (2,2) edge [out=10,in=80,distance=1cm] (2,2);
    \foreach \posx in {1, 2, 3}
    \foreach \posy in {1, 2}
    \draw[fill=white] (\posx, \posy) circle (1mm);;
    \draw (0,2) edge [out=10,in=80,distance=1cm] (0,2);
    \foreach \posy in {1, 2}
    \draw[fill=white] (0, \posy) circle (1mm);
    \draw (2,0) edge [out=-10,in=-80,distance=1cm] (2,0);
    \foreach \posx in {1, 2, 3}
    \draw[fill=white] (\posx, 0) circle (1mm);
\end{tikzpicture}}\bigskip

\subfigure[Strong ($G_1 \Strong G_2$)]{%
  \begin{tikzpicture}[scale=.8]
    \draw (1,1) grid (3,2); 
    \foreach \posx in {1, 2} {
      \foreach \posy in {1} {
        \draw (\posx, \posy) -- (\posx+1,\posy+1);
        \draw (\posx, \posy+1) -- (\posx+1,\posy);
      }
    }    
    \draw (1,0) grid (3,0); \node at (1.5,-0.5) {$G_2$};
    \draw (0,1) grid (0,2); \node at (-0.5,1.5) {$G_1$};
    \draw (1,2) edge [out=10,in=80,distance=1cm] (1,2);
    \draw (2,2) edge [out=10,in=80,distance=1cm] (2,2);
    \draw (3,2) edge [out=10,in=80,distance=1cm] (3,2);
    \draw (2,1) edge [out=-10,in=-80,distance=1cm] (2,1);
    \foreach \posx in {1, 2, 3}
    \foreach \posy in {1, 2}
    \draw[fill=white] (\posx, \posy) circle (1mm);;
    \draw (0,2) edge [out=10,in=80,distance=1cm] (0,2);
    \foreach \posy in {1, 2}
    \draw[fill=white] (0, \posy) circle (1mm);
    \draw (2,0) edge [out=-10,in=-80,distance=1cm] (2,0);
    \foreach \posx in {1, 2, 3}
    \draw[fill=white] (\posx, 0) circle (1mm);
\end{tikzpicture}}%
\qquad
\subfigure[Lexicographic ($G_1 \Lexicographic G_2$)]{%
  \begin{tikzpicture}[scale=.8]
    \draw (1,1) grid (3,2); 
    \foreach \posx in {1, 2} {
      \foreach \posy in {1} {
        \draw (\posx, \posy) -- (\posx+1,\posy+1);
        \draw (\posx, \posy+1) -- (\posx+1,\posy);
      }
    }
    \draw (1,1) -- (3,2);
    \draw (1,2) -- (3,1);
    \draw (1,2) edge [out=25,in=155,distance=0.5cm] (3,2);
    \draw (1,0) grid (3,0); \node at (1.5,-0.5) {$G_2$};
    \draw (0,1) grid (0,2); \node at (-0.5,1.5) {$G_1$};
    \draw (1,2) edge [out=10,in=80,distance=1cm] (1,2);
    \draw (2,2) edge [out=10,in=80,distance=1cm] (2,2);
    \draw (3,2) edge [out=10,in=80,distance=1cm] (3,2);
    \draw (2,1) edge [out=-10,in=-80,distance=1cm] (2,1);
    \foreach \posx in {1, 2, 3}
    \foreach \posy in {1, 2}
    \draw[fill=white] (\posx, \posy) circle (1mm);;
    \draw (0,2) edge [out=10,in=80,distance=1cm] (0,2);
    \foreach \posy in {1, 2}
    \draw[fill=white] (0, \posy) circle (1mm);
    \draw (2,0) edge [out=-10,in=-80,distance=1cm] (2,0);
    \foreach \posx in {1, 2, 3}
    \draw[fill=white] (\posx, 0) circle (1mm);
\end{tikzpicture}}
\caption{Example of the different types of graph products.}\label{fig:product-example}
\end{figure}
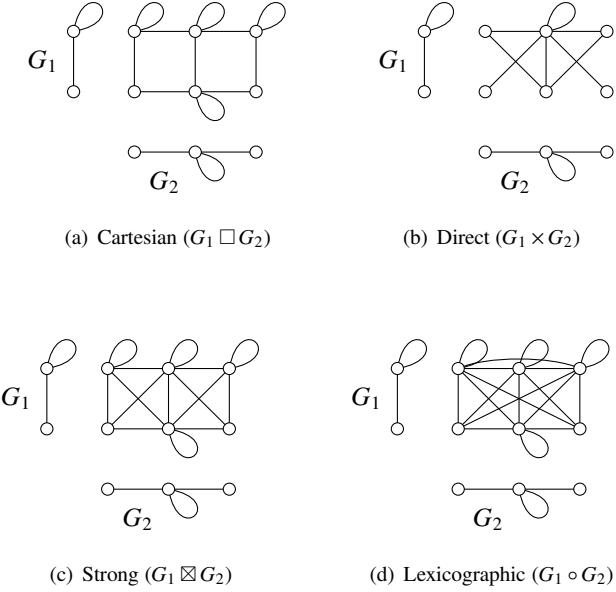

A nontrivial graph~$G \in \Gamma_0$ is a graph with more than one node
($|V(G)| > 1$). We say that a graph~$G$ is \emph{prime} according to a
given graph product $\odot$ if~$G$ is nontrivial and $G = G_1 \odot
G_2$ implies that either~$G_1$ or~$G_2$ are trivial, i.e., one of them
has exactly one node.

The direct product of $G_1, G_2$ can be specified in terms of the
\emph{Kronecker product} of their adjacency matrices. Given a 
$n \times m$ matrix~$\mat{A}$ and a $p \times q$ matrix~$\mat{B}$, the
Kronecker product $\mat{C} = \mat{A} \Kron \mat{B}$ is a $np \times
mq$ matrix obtained from the scalar multiplication between each
element of~$\mat{A}$ and the whole matrix~$\mat{B}$:

\begin{align*}
  \mat{C} &= \mat{A} \Kron \mat{B} = \begin{pmatrix}
    a_{11} \mat{B} & a_{12} \mat{B} & \ldots & a_{1m} \mat{B} \\
    a_{21} \mat{B} & a_{22} \mat{B} & \ldots & a_{2m} \mat{B} \\
    \vdots & \vdots & \ddots & \vdots \\
    a_{n1} \mat{B} & a_{n2} \mat{B} & \ldots & a_{nm} \mat{B} \\
    \end{pmatrix}
\end{align*}

It can be easily shown~\cite{hammack2011handbook} that the adjacency
matrix~$\adj{G}$ of the graph $G = G_1 \Direct G_2$ is strongly related 
 (see Observation~\ref{obs:kron-mat-kron-graph}) to the
adjacency matrices~$\adj{G_1}$ and~$\adj{G_2}$.

We finally recall the definition of \emph{many-one reducibility} and
\emph{Turing reducibility}. Given two sets $S_1, S_2 \subseteq \N$, we
say that~$S_1$ is many-one reducible to~$S_2$, if
there exists a total computable function $f: \N \rightarrow \N$ such
that~\cite{post1944}

\begin{align*}
n \in S_1 \iff f(n) \in S_2
\end{align*}
A polynomial time many-one reduction (denoted by $\mreduce$) is a many-one reduction 
with the additional constraint that~$f$ is computable in polynomial time.

Turing reducibility is a weaker form of many-one
reducibility. Informally, $S_1$ is Turing reducible to~$S_2$ if there exists an oracle for testing membership in~$S_1$ relying on another oracle for testing membership in
$S_2$~\cite{hartley1967}. In other words, $S_1$ is Turing reducible
to~$S_2$ if it is possible to answer the question ``is $n \in S_1$''
given the existence of an effective procedure for answering the
question ``is $m \in S_2$'' for any $m \in \N$~\cite{post1944}.

A polynomial time Turing reduction (denoted by $\treduce$) is a Turing reduction satisfying the following two additional constraints:
\begin{enumerate}
\item the oracle for testing membership in~$S_1$ makes at most a
  polynomial number of calls to the oracle for testing membership
  in~$S_2$ and
\item the overall computational cost of the oracle for testing
  membership in~$S_1$ (excluding the calls to the oracle for testing
  membership in~$S_2$) is polynomially bounded.
\end{enumerate}

\begin{table}[t]
  \centering%
  \begin{tabular}{ll}
    \toprule
    \textbf{Symbol} & \textbf{Description} \\
    \midrule
    $\Gamma$    & The set of finite, undirected graphs, without self-loops\\
    $\Gamma_0$  & The set of finite, undirected graphs, self-loops allowed\\
    $\adj{G}$ & The adjacency matrix of a graph $G$ \\
    $\mat{I}_n$ & The $n \times n$ identity matrix \\
    $\mat{0}_n$ & The $n \times n$ zero matrix \\
    $\Direct$   & The direct graph product operator \\
    $\Cartesian$   & The Cartesian graph product operator \\
   $\Strong$   & The strong graph product operator \\
   $\Lexicographic$   & The lexicographic graph product operator \\
    $\Kron$   & The Kronecker matrix product operator \\
     $\Isomorph$ & The graphs isomorphism operator \\
    $\mreduce$  & Polynomial many-one reducibility \\
    $\treduce$  & Polynomial Turing reducibility \\
    $\cup$  & Disjoint union of graphs \\
  \bottomrule
  \end{tabular}
  \caption{Basic notation.}\label{tab:notation}
\end{table}

As a final consideration, throughout the paper we intend graphs to be \emph{finite} and \emph{undirected}, unless otherwise specified.
Table~\ref{tab:notation} summarizes the notation used in this paper.
\section{Related works}\label{sec:related}

In this section we list the main
results on graph factorization that 
are strictly related to the work presented in this paper.
The interested reader may find a comprehensive review of the 
theory of graph factorizations in the recent book
by Hammack, Imrich and Klav{\v{z}}ar~\cite{hammack2011handbook}.

\paragraph{Direct product}
Prime factorization of connected, nonbipartite graphs in $\Gamma_0$ is unique 
up to isomorphism and the order of the factors, and
can be computed in polynomial time~\cite{imrich1998}.

\paragraph{Cartesian product}
Prime factorization of connected 
graphs is unique 
up to isomorphism and the order of the factors~\cite{Sabidussi1959/60, MR0209178}.
Prime factorization is not unique in the class of possibly disconnected simple graphs.
Following Sabidussi's approach, Feigenbaum et al.~\cite{feigenbaum1985} 
derived a polynomial-time algorithm that computes the prime factors of a connected graph.
A different polyonimial-time algorithm for connected graphs has been independently
discovered by Winkler~\cite{winkler1987}.

\paragraph{Strong product}
Prime factorization of connected graphs is unique up to 
reorderings and isomorphisms of factors and it
can be computed in polynomial time~\cite{imrich1998}.

\paragraph{Lexicographic product}
Determining whether a connected graph is prime 
is at least as difficult as the graph isomorphism problem~\cite{feigenbaum1986}.\\

\begin{table}[t]
\begin{center}
  \begin{tabular}{rcccc}
    \toprule
    {\bf Graph type} & \multicolumn{4}{c}{\bf Product type} \\
    & Direct
     & Cartesian
      & Strong
       & Lexicographic \\
    \cmidrule(l){2-5}
    Connected, nonbipartite
    & \textbf{P}~\cite{imrich1998}
     & \color{lightgray}P~\cite{feigenbaum1985,winkler1987}
      & \color{lightgray}P~\cite{feigenbaum1992}
       & $\bullet$  \\
    Connected
    & $\bullet$
     & \textbf{P}~\cite{feigenbaum1985,winkler1987}
      & \textbf{P}~\cite{feigenbaum1992}
       & \textbf{GI-Hard}~\cite{feigenbaum1986} \\
    Unconnected, nonbipartite
    & \textbf{GI-Hard} (our results)
     & $\bullet$
      & $\bullet$
       & $\bullet$ \\
    Nonbipartite
    & \color{lightgray}\textbf{GI-Hard} (our results)
     & $\bullet$
      & $\bullet$
       & $\bullet$ \\
    \bottomrule
  \end{tabular}
\end{center}
\caption{Complexity of the
graph factorization problem for  different types of graphs 
considered in the literature (connected, unconnected, nonbipartite) and different notions of
graph product (direct, cartesian, strong and lexicographic product);
 \textbf{P} stands for polynomially time solvable. Table cells reported in light gray can be easily inferred
 from another cells in the same column depending on the relation between the corresponding classes of graphs. For instance, a polynomial-time solvable problem for connected graphs is polynomial-time solvable within a restricted class of graphs (e.g., connected and nonbipartite). Dots denote the cases that, to our knowledge, have not yet been explored.
  }\label{tab:previous-results}
\label{default}
\end{table}

An interesting observation relating graph factorization and graph isomorphism problem
can be found at the end of page of~\cite[p. 229]{hammack2011handbook}.
The authors claim that, if~$X$ is the disjoint union of graphs~$G$ and~$H$,
then $G \Isomorph H$ if and only if $X = D_2 \Cartesian G = D_2 \Strong G = D_2 \Lexicographic G$
where~$D_2$ denotes the graph with two nodes and two self-loops.
They conclude that ``testing whether a disconnected graph is decomposable with respect to any 
of these three products is at least as hard as the graph isomorphism problem''. 
They do not give a formal proof of their claim and in particular they do not explain how they get rid of the 
case in which~$X$ is the disjoint union of two non isomorphic graphs~$G_1$ and~$G_2$ and, at the same time,  $X$
admits as a factor a graph with two nodes that is not isomorphic to~$D_2$.
In that case the decomposability test would lead to erroneously declare $G_1 \simeq G_2$.
Moreover, if~$X$ admits more than one factorization, also computing a single factorization could not be enough for 
testing isomorphism.

In Section~\ref{sec:results}  we show (see Figures~\ref{fig:non-unique-factorization} and~\ref{fig:whichever-factorization}) 
that when the direct product is used, both these cases
may occur.
\section{Main Results}\label{sec:results}

In this section we prove that testing whether two graphs $G_1, G_2$
are isomorphic is not harder than testing whether an undirected
graph $G \in \Gamma_0$ is $\Direct$-composite, i.e., $G$ admits
nontrivial factors with respect to the direct product decomposition. 
More formally, we show
that graph isomorphism problem is polynomially many-one
reducible to the problem of testing $\Direct$-compositeness of graphs.

Before starting, we formally define the following three decision problems in terms of their admissible inputs 
and related outputs.

\begin{defi}[\gi{$S\hspace{-0.03cm}$}] Let~$S$ be any set of graphs. \gi{$S\hspace{-0.03cm}$} is defined as follows.\\
Input: $G_1, G_2 \in S$\\
Output: {\ansyes} if~$G_1$ is isomorphic to~$G_2$, {\ansno} otherwise.
\end{defi}

\begin{defi}[\primality{$S\hspace{-0.03cm}$}]
Let~$S$ be any set of graphs. \primality{$S\hspace{-0.03cm}$} is defined as follows.\\
Input: $G \in S$\\
Output: {\ansyes} if~$G$ is prime with respect
 to the direct product, {\ansno} otherwise.
\end{defi}

\begin{defi}[\compositeness{$S\hspace{-0.03cm}$}]
Let~$S$ be any set of graphs. \compositeness{$S\hspace{-0.03cm}$} is defined as follows.\\
Input: $G \in S$\\
Output: {\ansyes} if~$G$ is composite with respect
 to the direct product, {\ansno} otherwise.
\end{defi}

\begin{defi}[$GI$-hard problem]
A decisional problem~$P$ is $GI$-hard if and only if
$$\gi{general\ graphs\hspace{-0.03cm}} \treduce P$$
\end{defi}

It is easy to observe that $\gi{general\ graphs\hspace{-0.03cm}} \treduce \gi{connected\ graphs\hspace{-0.03cm}}$
so that, by transitivity, we can conclude that a decisional problem~$P$ is $GI$-hard if and only if
$$\gi{connected\ graphs\hspace{-0.03cm}} \treduce P$$

The following observation 
highlights the  strong relation
between the direct product of graphs and the Kronecker product of their adjacency matrices.

\begin{obs}\label{obs:kron-mat-kron-graph}
Let~$G_1$ and~$G_2$ be graphs.
Then
\begin{equation*}
\adj{G_1 \Direct G_2} =   \tr{\mat{P}}\,(\adj{G_1} \Kron \adj{G_2})\,\mat{P}  ,
\end{equation*}
where $\mat{P}$ is a suitable permutation matrix.
\end{obs}

In the following lemma we prove that  two graphs~$G_1$ and~$G_2$
are isomorphic if and only if  there exists a permutation matrix~$\mat{P}$  that transforms
the adjacency matrix of the disjoint union of~$G_1$ and~$G_2$ into
the Kronecker product of the identity matrix~$\mat{I}_2$ and a suitable binary
matrix~$\mat{B}$.

\begin{lemma}\label{g1-iso-g2}
Let $G_1, G_2$ be undirected, connected graphs with~$n$ nodes. 
Let $\mat{M}_1=\adj{G_1}$ and $\mat{M}_2=\adj{G_2}$.
Let $\mat{M}$ denote the adjacency matrix of the disjoint union $G = G_1 \cup G_2$. 
Without loss of generality, we may write $\mat{M}$ as
  
  \begin{equation}
    \mat{M} = \begin{pmatrix}
      \mat{M}_1 & \mat{0}_n\\
      \mat{0}_n & \mat{M}_2
      \end{pmatrix}\label{eq:g1-iso-g2-M}
  \end{equation}
  
  \noindent Then, $G_1$ and~$G_2$ are isomorphic ($G_1 \Isomorph G_2$) if
  and only if there exists a $2n \times 2n$ permutation
  matrix~$\mat{P}$ and a $n \times n$ binary matrix~$\mat{B}$ such
  that

  \begin{equation}
    \tr{\mat{P}} \mat{M} \mat{P} = \mat{I}_2 \Kron \mat{B}\label{eq:g1-iso-g2-2}
  \end{equation}
  where~$I_2$ denotes the $2\times 2$ identity matrix.
\end{lemma}

\begin{proof}
  ($\implies$) We first prove that if~$G_1$ and~$G_2$ are isomorphic, then
  Eq.~\eqref{eq:g1-iso-g2-2} holds. If $G_1 \Isomorph G_2$ then there
  exists a $n \times n$ permutation matrix~$\mat{Q}$ that transforms the
  adjacency matrix~$\mat{M}_2$ of~$G_2$ in the adjacency matrix~$\mat{M}_1$ of~$G_1$:

  \begin{equation}
    \mat{M}_1 = \tr{\mat{Q}} \mat{M}_2 \mat{Q}\label{eq:g1-iso-g2-M1}
  \end{equation}

  \noindent Let us define

  \begin{equation}
    \mat{P} = \begin{pmatrix}
      \mat{I}_n & \mat{0}_n \\
      \mat{0}_n & \mat{Q}
    \end{pmatrix}\label{eq:g1-iso-g2-P}
  \end{equation}

  \noindent It follows that
  
  \begin{align*}
    \tr{\mat{P}} \mat{M} \mat{P} &=
    \begin{pmatrix}
      \mat{I}_n & \mat{0}_n\\
      \mat{0}_n & \tr{\mat{Q}}
    \end{pmatrix}
    \begin{pmatrix}
        \mat{M}_1 & \mat{0}_n\\
        \mat{0}_n & \mat{M}_2
    \end{pmatrix}
    \begin{pmatrix}
        \mat{I}_n & \mat{0}_n\\
        \mat{0}_n & \mat{Q}
    \end{pmatrix} & \text{by~\eqref{eq:g1-iso-g2-P} and~\eqref{eq:g1-iso-g2-M}}\\
    &= \begin{pmatrix}
      \mat{M}_1 & \mat{0}_n\\
      \mat{0}_n & \tr{\mat{Q}} \mat{M}_2 \mat{Q}
    \end{pmatrix}\\    
    &= \begin{pmatrix}
        \mat{M}_1 & \mat{0}_n\\
        \mat{0}_n & \mat{M}_1
      \end{pmatrix} & \text{by~\eqref{eq:g1-iso-g2-M1}} \\
    &= \mat{I}_2 \Kron \mat{M}_1
  \end{align*}
  
  \noindent ($\Longleftarrow$) We prove that if Eq.~\eqref{eq:g1-iso-g2-2}
  holds, then $G_1, G_2$ are isomorphic.

  \noindent Observe that the transformation $\tr{\mat{P}} \mat{M} \mat{P}$
  consists of relabeling the nodes of~$G$ according to the
  permutation matrix~$\mat{P}$. Let us define this relabeling as the
  bijective function $\pi: \{1, 2, \ldots, 2n\} \to \{1, 2, \ldots,
  2n\}$. Therefore, $\pi(i) = j$ if and only if the node~$i$ of~$G$ is relabeled
  as~$j$. Note also that $\tr{\mat{P}} \mat{M} \mat{P}$ is symmetric,
  since it represents the adjacency matrix of an undirected graph.
  
  \noindent Since we are assuming that~$G_1$ is connected, then there always
  exists a path from node~$i$ to node~$j$, $1 \leq i,j \leq n$. Thus,
  should the permutation contain a mapping such that $\pi(i) \leq n$
  and $\pi(j) > n$, the relabeled adjacency matrix $\tr{\mat{P}}
  \mat{M} \mat{P}$ would contain at least one~$1$ in the upper-right
  quadrant and (by symmetry) in the lower-left one. However, this
  contradicts the hypothesis~\eqref{eq:g1-iso-g2-2}, since the upper
  right and lower left quadrants of $\mat{I}_2 \Kron \mat{B}$ are the
  zero matrix $\mat{0}_n$. The same considerations apply to~$G_2$.

  \noindent Thus, from Eq.~\eqref{eq:g1-iso-g2-2} we observe that~$\pi$ maps the sets
  $\{1, 2, \ldots, n\}$ and $\{n+1, n+2, \ldots, 2n\}$ into
  themselves, and therefore~$\mat{P}$ must have a block structure:
  
  \begin{equation}
    \mat{P} = \begin{pmatrix}
      \mat{P}_1 & \mat{0}_n \\
      \mat{0}_n & \mat{P}_2
    \end{pmatrix}\label{eq:g1-iso-g2-P1P2}
  \end{equation}

  \noindent Let~$G_3$ be the undirected graph such that~$\adj{G_3}=\mat{B}$.
  Combining~\eqref{eq:g1-iso-g2-P1P2} and~\eqref{eq:g1-iso-g2-2} we
  can conclude that~$G_1 \Isomorph G_3$ and $G_2 \Isomorph G_3$,
  because the adjacency matrices $\mat{M}_1, \mat{M}_2$ can be
  transformed into~$\mat{B}$ via the permutation matrices $\mat{P}_1,
  \mat{P}_2$, respectively. By transitivity we conclude $G_1 \Isomorph G_2$.
\end{proof}

Lemma \ref{g1-iso-g2} ensures
that the adjacency matrix of the disjoint union of two isomorphic graphs may always be written as $\mat{I}_2 \Kron \mat{B}$; note that $\mat{I}_2$ is the adjacency matrix of $D_2$, the graph with two nodes and two self-loops. 
Unfortunately, simply testing primality of the disjoint union~$X = G_1 \cup G_2$ of two graphs~$G_1$ and~$G_2$ is not enough for 
deciding whether~$G_1$ and~$G_2$ are isomorphic or not.
In fact, as mentioned at the end of Section~\ref{sec:related}, the graph~$X$ could admit as a factor a graph with two nodes
different from~$D_2$.
For example, Figure~\ref{fig:whichever-factorization} shows a graph~$X$ such that
\begin{itemize}
\item $X$ is the disjoint union of two non isomorphic graphs (connected and having the same number 
of nodes and edges) and
\item $X$ admits as a factor a graph with two nodes
  different from~$D_2$.
\end{itemize}

Moreover, the idea of factorizing $X=G_1 \cup G_2$ to check whether~$D_2$ is a factor 
(or, equivalently, $G_1 \simeq G_2$) 
might fail due to the fact that~$X$ could admit two different factorizations~$F_1$ and~$F_2$ where~$F_1$ contains~$D_2$  while~$F_2$ does not. 
Figure~\ref{fig:non-unique-factorization} shows an example of a graph~$X$ such that
\begin{itemize}
\item $X$ is the disjoint union of two isomorphic graphs (both connected and with a prime number of nodes)
\item $X$ admits two distinct factorizations~$F_1$ and~$F_2$ where~$F_1$ contains~$D_2$ while~$F_2$ does not.
\end{itemize}

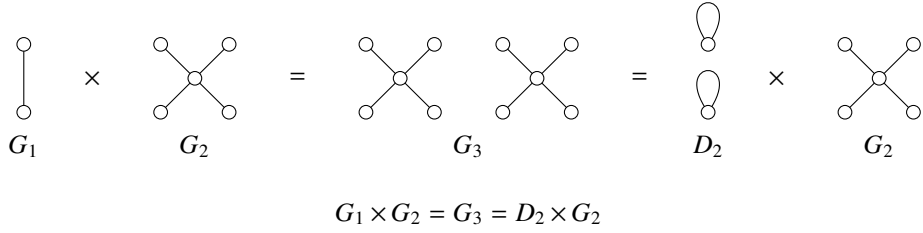
\begin{figure}[t]
\centering
\begin{tikzpicture}[scale=.9]
    \node at (0,-0.5) {$G_1$};
    \draw (0,0) grid (0,1);
    \foreach \posy in {0, 1}
    \draw[fill=white] (0, \posy) circle (1mm);
    \node at (1,0.5) {$\Direct$};
    \node at (2.5,-0.5) {$G_2$};
    \draw (2,0) -- (3,1);
    \draw (2,1) -- (3,0);    
    \foreach \posx in {2, 3}
    \foreach \posy in {0, 1}
    \draw[fill=white] (\posx, \posy) circle (1mm);;
    \draw[fill=white] (2.5, 0.5) circle (1mm);    
    \node at (4,0.5) {$=$};
    \node at (6.5,-0.5) {$G_3$};
    \draw (5,0) -- (6,1);
    \draw (5,1) -- (6,0);
    \draw (7,0) -- (8,1);
    \draw (7,1) -- (8,0);
    \foreach \posx in {5, 6, 7, 8}
    \foreach \posy in {0, 1}
    \draw[fill=white] (\posx, \posy) circle (1mm);;
    \draw[fill=white] (5.5, 0.5) circle (1mm);
    \draw[fill=white] (7.5, 0.5) circle (1mm);
	\node at (9,0.5) {$=$};    
    \node at (10,-0.5) {$D_2$};
    \draw (10,0) edge [out=55,in=125,distance=1cm] (10,0);
    \draw (10,1) edge [out=55,in=125,distance=1cm] (10,1);
    \foreach \posy in {0, 1}
    \draw[fill=white] (10, \posy) circle (1mm);
    \node at (11,0.5) {$\Direct$};
    \node at (12.5,-0.5) {$G_2$};
    \draw (12,0) -- (13,1);
    \draw (12,1) -- (13,0);    
    \foreach \posx in {12, 13}
    \foreach \posy in {0, 1}
    \draw[fill=white] (\posx, \posy) circle (1mm);;
    \draw[fill=white] (12.5, 0.5) circle (1mm);    

    \node at (6.5,-1.5) {$G_1 \Direct G_2 = G_3 = D_2 \Direct G_2$};

\end{tikzpicture}
\caption{
$G_2$ is a connected graph with a prime number of nodes. 
$G_3$ is the disjoint union of two copies of $G_2$. $G_3$ admits two different factorizations, namely, $G_1 \Direct G_2$ 
and $D_2 \Direct G_2$. }
\label{fig:non-unique-factorization}
\end{figure}

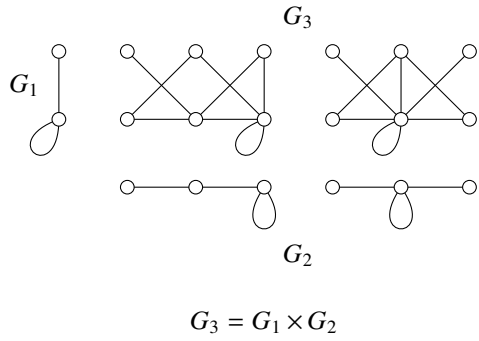
\begin{figure}[t]
\centering
\begin{tikzpicture}[scale=.9]
   \node at (3.5,2.5) {$G_3$};
    \draw (2,1) -- (1,2);
    \draw (2,1) -- (3,2);
    \draw (2,2) -- (1,1);
    \draw (2,2) -- (3,1);
    \draw (3,2) -- (3,1);
    \draw (1,1) grid (3,1);
    \draw (5,1) -- (4,2);
    \draw (5,1) -- (6,2);
    \draw (5,2) -- (4,1);
    \draw (5,2) -- (5,1);
    \draw (5,2) -- (6,1);
    \draw (4,1) grid (6,1);
    \draw (3,1) edge [out=200,in=270,distance=1cm] (3,1);
    \draw (5,1) edge [out=200,in=270,distance=1cm] (5,1);
    \foreach \posx in {1, 2, 3, 4, 5, 6}
    \foreach \posy in {1, 2}
    \draw[fill=white] (\posx, \posy) circle (1mm);;
    \node at (-0.5,1.5) {$G_1$};
    \draw (0,1) grid (0,2);
    \draw (0,1) edge [out=200,in=270,distance=1cm] (0,1);
    \foreach \posy in {1, 2}
    \draw[fill=white] (0, \posy) circle (1mm);
   \node at (3.5,-1) {$G_2$};
    \draw (1,0) grid (3,0);
    \draw (4,0) grid (6,0);
    \draw (3,0) edge [out=235,in=305,distance=1cm] (3,0);
    \draw (5,0) edge [out=235,in=305,distance=1cm] (5,0);
    \foreach \posx in {1, 2, 3, 4, 5, 6}
    \draw[fill=white] (\posx, 0) circle (1mm);

\node at (3,-2) {$G_3 = G_1 \Direct G_2$};

\end{tikzpicture}
\caption{
$G_3$ is the disjoint union of two connected graphs with the same number of nodes and edges.
$G_3$ is the is the direct  product of $G_1$ and $G_2$.  $G_3$ admits as a factor $G_1$ which is a graph with two nodes different from $D_2$.
$G_3$ does not admit $D_2$ as a factor.}
\label{fig:whichever-factorization}
\end{figure}

In order to prove the main result of this paper, we define a class of graphs~$\GG$ as follows.

\begin{defi}[Class~$\GG$]
  A graph $G = (V, E)$ with~$n$ nodes and~$m$ edges belongs to the class $\GG \subset \Gamma_0$ if and only if:

  \begin{enumerate}[label={\bfseries P\arabic*}]
    \item $G$ is undirected, connected and not bipartite;\label{prop:bipartite}
    \item The number of nodes~$n$ is prime;\label{prop:prime}
    \item The number~$s$ of self-loops is strictly less than the number of edges, i.e.,  $s < m$;
    \item $(2m - s)$ is not divisible by~$2$;\label{prop:div2}
    \item $(2m - s)$ is not divisible by~$3$;\label{prop:div3}
  \end{enumerate}
\end{defi}

In the following theorem we give a polynomial time many-one reduction from 
$\gi{\GG}$ to $\compositeness{\GG}$.
A consequence of this result is that
the existence of a polynomial-time algorithm to determine whether a
  given graph in~$\GG$ is composite with respect to the direct product would imply the existence of
  a polynomial-time algorithm for the graph isomorphism problem
  between any two graphs in~$\GG$.
  
\begin{teo}\label{classG}
  \begin{equation*}
    \gi{\GG} \mreduce  \compositeness{\GG}
  \end{equation*}
\end{teo}

\begin{proof}
Let ${\cal{A}}$ be an algorithm
  for testing compositeness for graphs in~$\GG$.
   The following algorithm solves the
   isomorphism problem for graphs in~$\GG$ relying on a single call of ${\cal{A}}$,
  therefore providing a polynomial time many-one reduction 
  from $\gi{\GG}$ to $\compositeness{\GG}$.
  
\begin{codebox}
\Procname{$\GG\proc{-Graph-Isomorphism}\left(G_1, G_2\right)$}
 \li 	\If $|V(G_1)| \neq |V(G_2)|$ {\bf or} $|E(G_1)| \neq |E(G_2)|$
 \li 	\Then \Return \textsc{no} \End
 \li 	$G \gets G_1 \cup G_2$ \Comment graphs disjoint union
\li \Return ${\cal{A}}(G)$
\end{codebox}
 
\noindent Let us prove that  $\GG\proc{-Graph-Isomorphism}$ is correct.
To this end, we consider two cases:

\paragraph*{$G$ is prime}
Let $G_1, G_2 \in \GG$ and $\mat{M}=\adj{G_1 \cup G_2}$. 
According to Lemma~\ref{g1-iso-g2} and to Observation~\ref{obs:kron-mat-kron-graph}, if~$G$ is prime we can conclude
that~$G_1$ and~$G_2$ are not isomorphic, since if they were, there
should exist a permutation matrix~$\mat{P}$ and a suitable adjacency matrix~$\mat{B}$ 
such that $\tr{\mat{P}} \mat{M} \mat{P} =
\mat{I}_2 \Kron \mat{B}$ and then $G_1 \cup G_2$ would be composite. 

\paragraph*{$G$ is composite}
Let $G_1, G_2 \in \GG$ and $\mat{M}=\adj{G_1 \cup G_2}$. 
Let $n = |V(G_1)| = |V(G_2)|$ be the number of nodes of either $G_1$
or $G_2$. Since $G_1 \in \GG$, $n$ is prime
(\ref{prop:prime}). Consequently, the number of nodes of~$G = G_1 \cup
G_2$ is~$2n$ and its adjacency matrix~$\mat{M}$ has size $2n \times
2n$. Therefore, the only possible factorization of~$\mat{M}$ is
$\mat{M} = \mat{A} \Kron \mat{B}$, where~$\mat{A}$ has size $2 \times 2$ and~$\mat{B}$ has
size $n \times n$. Additionally, since $G_1, G_2 \in \GG$,
their adjacency matrices have exactly $2m-s$ nonzero elements each,
where $m=|E(G_1)|=|E(G_2)|$ is the number of edges of either $G_1$ or
$G_2$ and~$s$ is the number of self-loops. Let us consider each
possible configuration of the matrix~$\mat{A}$:

$$
\begin{pmatrix}0&0\\0&0\end{pmatrix}\quad
\begin{pmatrix}1&0\\0&0\end{pmatrix}\quad
\begin{pmatrix}0&1\\0&0\end{pmatrix}\quad
\begin{pmatrix}0&0\\1&0\end{pmatrix}\quad
\begin{pmatrix}0&0\\0&1\end{pmatrix}\quad
\begin{pmatrix}1&1\\0&0\end{pmatrix}\quad
\begin{pmatrix}1&0\\1&0\end{pmatrix}\quad
\begin{pmatrix}1&0\\0&1\end{pmatrix}
$$
$$
\begin{pmatrix}0&1\\1&0\end{pmatrix}\quad
\begin{pmatrix}0&1\\0&1\end{pmatrix}\quad
\begin{pmatrix}0&0\\1&1\end{pmatrix}\quad
\begin{pmatrix}0&1\\1&1\end{pmatrix}\quad
\begin{pmatrix}1&0\\1&1\end{pmatrix}\quad
\begin{pmatrix}1&1\\0&1\end{pmatrix}\quad
\begin{pmatrix}1&1\\1&0\end{pmatrix}\quad
\begin{pmatrix}1&1\\1&1\end{pmatrix}
$$

\noindent
Since~$G_1$ and~$G_2$ are undirected, $G$ is undirected as well, so
its adjacency matrix~$\mat{M}$ must be symmetric. From~\ref{g1-iso-g2}
we deduce that~$\mat{A}$ must be symmetric, since~$\mat{B}$ is a
nonzero matrix. Therefore, we exclude all configurations of matrix~$\mat{A}$
that are not symmetrix.

$$
\begin{pmatrix}0&0\\0&0\end{pmatrix}\quad
\begin{pmatrix}1&0\\0&0\end{pmatrix}\quad
\color{lightgray}
\begin{pmatrix}0&1\\0&0\end{pmatrix}\quad
\begin{pmatrix}0&0\\1&0\end{pmatrix}\quad
\color{black}
\begin{pmatrix}0&0\\0&1\end{pmatrix}\quad
\color{lightgray}
\begin{pmatrix}1&1\\0&0\end{pmatrix}\quad
\begin{pmatrix}1&0\\1&0\end{pmatrix}\quad
\color{black}
\begin{pmatrix}1&0\\0&1\end{pmatrix}
$$
$$
\begin{pmatrix}0&1\\1&0\end{pmatrix}\quad
\color{lightgray}
\begin{pmatrix}0&1\\0&1\end{pmatrix}\quad
\begin{pmatrix}0&0\\1&1\end{pmatrix}\quad
\color{black}
\begin{pmatrix}0&1\\1&1\end{pmatrix}\quad
\color{lightgray}
\begin{pmatrix}1&0\\1&1\end{pmatrix}\quad
\begin{pmatrix}1&1\\0&1\end{pmatrix}\quad
\color{black}
\begin{pmatrix}1&1\\1&0\end{pmatrix}\quad
\begin{pmatrix}1&1\\1&1\end{pmatrix}
$$

\noindent
Since~$G_1$ and~$G_2$ are connected, the degrees of all their nodes
must be strictly grater than zero. Therefore, $\mat{A}$ must contain
more than a single nonzero element as, conversely, the resulting
matrix~$\mat{M}$ would have at least~$n$ unconnected nodes with zero
degree. Therefore, we exclude all configurations of~$\mat{A}$ that have
less than two nonzero elements.

$$
\color{lightgray}
\begin{pmatrix}0&0\\0&0\end{pmatrix}\quad
\begin{pmatrix}1&0\\0&0\end{pmatrix}\quad
\begin{pmatrix}0&1\\0&0\end{pmatrix}\quad
\begin{pmatrix}0&0\\1&0\end{pmatrix}\quad
\begin{pmatrix}0&0\\0&1\end{pmatrix}\quad
\begin{pmatrix}1&1\\0&0\end{pmatrix}\quad
\begin{pmatrix}1&0\\1&0\end{pmatrix}\quad
\color{black}
\begin{pmatrix}1&0\\0&1\end{pmatrix}
$$
$$
\begin{pmatrix}0&1\\1&0\end{pmatrix}\quad
\color{lightgray}
\begin{pmatrix}0&1\\0&1\end{pmatrix}\quad
\begin{pmatrix}0&0\\1&1\end{pmatrix}\quad
\color{black}
\begin{pmatrix}0&1\\1&1\end{pmatrix}\quad
\color{lightgray}
\begin{pmatrix}1&0\\1&1\end{pmatrix}\quad
\begin{pmatrix}1&1\\0&1\end{pmatrix}\quad
\color{black}
\begin{pmatrix}1&1\\1&0\end{pmatrix}\quad
\begin{pmatrix}1&1\\1&1\end{pmatrix}
$$

\noindent
Let us now consider the number~$b$ of nonzero elements in
$\mat{B}$. Should~$\mat{A}$ have three nonzero elements, the resulting
number of nonzero elements in~$\mat{M}$ would be~$3b$, which is
divisible by~$3$. Moreover, the number of nonzero elements
of~$\mat{M}$ is equal to the number of nonzero elements of the
adjacency matrices of~$G_1$ and~$G_2$:

\begin{equation}
  3b = 2(2m-s)\label{eq:3b}
\end{equation}

\noindent
Since $G_1, G_2 \in \GG$, we know that the number of nonzero
elements ($2m-s$) in their adjacency matrices must not be divisible
by~$3$ (\ref{prop:div3}). Thus, $2(2m-s)$ must not be divisible by~$3$
either, contradicting~\eqref{eq:3b}. We conclude that~$\mat{A}$ can
not contain three nonzero elements.

$$
\color{lightgray}
\begin{pmatrix}0&0\\0&0\end{pmatrix}\quad
\begin{pmatrix}1&0\\0&0\end{pmatrix}\quad
\begin{pmatrix}0&1\\0&0\end{pmatrix}\quad
\begin{pmatrix}0&0\\1&0\end{pmatrix}\quad
\begin{pmatrix}0&0\\0&1\end{pmatrix}\quad
\begin{pmatrix}1&1\\0&0\end{pmatrix}\quad
\begin{pmatrix}1&0\\1&0\end{pmatrix}\quad
\color{black}
\begin{pmatrix}1&0\\0&1\end{pmatrix}
$$
$$
\begin{pmatrix}0&1\\1&0\end{pmatrix}\quad
\color{lightgray}
\begin{pmatrix}0&1\\0&1\end{pmatrix}\quad
\begin{pmatrix}0&0\\1&1\end{pmatrix}\quad
\begin{pmatrix}0&1\\1&1\end{pmatrix}\quad
\begin{pmatrix}1&0\\1&1\end{pmatrix}\quad
\begin{pmatrix}1&1\\0&1\end{pmatrix}\quad
\begin{pmatrix}1&1\\1&0\end{pmatrix}\quad
\color{black}
\begin{pmatrix}1&1\\1&1\end{pmatrix}
$$

\noindent
Similarly, should the matrix~$\mat{A}$ have four nonzero elements, the
number of nonzero elements in~$\mat{M}$ would be $4b$; from the
same reasoning above, we get:

\begin{equation}
  4b = 2(2m-s)
\end{equation}

\noindent
However, from~\ref{prop:div2} we have that $(2m-s)$ is not divisible
by~$2$, and therefore $2(2m-s)$ is not divisible by~$4$. We conclude
that~$\mat{A}$ can not have four nonzero elements.

$$
\color{lightgray}
\begin{pmatrix}0&0\\0&0\end{pmatrix}\quad
\begin{pmatrix}1&0\\0&0\end{pmatrix}\quad
\begin{pmatrix}0&1\\0&0\end{pmatrix}\quad
\begin{pmatrix}0&0\\1&0\end{pmatrix}\quad
\begin{pmatrix}0&0\\0&1\end{pmatrix}\quad
\begin{pmatrix}1&1\\0&0\end{pmatrix}\quad
\begin{pmatrix}1&0\\1&0\end{pmatrix}\quad
\color{black}
\begin{pmatrix}1&0\\0&1\end{pmatrix}
$$
$$
\begin{pmatrix}0&1\\1&0\end{pmatrix}\quad
\color{lightgray}
\begin{pmatrix}0&1\\0&1\end{pmatrix}\quad
\begin{pmatrix}0&0\\1&1\end{pmatrix}\quad
\begin{pmatrix}0&1\\1&1\end{pmatrix}\quad
\begin{pmatrix}1&0\\1&1\end{pmatrix}\quad
\begin{pmatrix}1&1\\0&1\end{pmatrix}\quad
\begin{pmatrix}1&1\\1&0\end{pmatrix}\quad
\begin{pmatrix}1&1\\1&1\end{pmatrix}
$$

\noindent
Finally, we point out that since~$G_1$ and~$G_2$ are not
bipartite (\ref{prop:bipartite}), $G$ is not bipartite as
well. Should~$\mat{A} = \begin{pmatrix}0&1\\1&0\end{pmatrix}$, matrix
$\mat{M} = \mat{A} \Kron \mat{B}$ would represent a bipartite graph,
where the first~$n$ nodes are only connected to the other~$n$ nodes
and viceversa. Therefore, $\mat{A}$ can not be in that form.

$$
\color{lightgray}
\begin{pmatrix}0&0\\0&0\end{pmatrix}\quad
\begin{pmatrix}1&0\\0&0\end{pmatrix}\quad
\begin{pmatrix}0&1\\0&0\end{pmatrix}\quad
\begin{pmatrix}0&0\\1&0\end{pmatrix}\quad
\begin{pmatrix}0&0\\0&1\end{pmatrix}\quad
\begin{pmatrix}1&1\\0&0\end{pmatrix}\quad
\begin{pmatrix}1&0\\1&0\end{pmatrix}\quad
\color{black}
\begin{pmatrix}1&0\\0&1\end{pmatrix}
$$
$$
\color{lightgray}
\begin{pmatrix}0&1\\1&0\end{pmatrix}\quad
\begin{pmatrix}0&1\\0&1\end{pmatrix}\quad
\begin{pmatrix}0&0\\1&1\end{pmatrix}\quad
\begin{pmatrix}0&1\\1&1\end{pmatrix}\quad
\begin{pmatrix}1&0\\1&1\end{pmatrix}\quad
\begin{pmatrix}1&1\\0&1\end{pmatrix}\quad
\begin{pmatrix}1&1\\1&0\end{pmatrix}\quad
\begin{pmatrix}1&1\\1&1\end{pmatrix}
$$

\noindent
We conclude that~$\mat{A} = \mat{I}_2$. Thus, according to
Lemma~\ref{g1-iso-g2}, $G_1$ and~$G_2$ are isomorphic.
\end{proof}

Theorem~\ref{classG} shows that, within the class~$\GG$,  there exists an intimate relation between
graph primality with respect to the direct product and graph isomorphism. 
In what follows we will extend Theorem~\ref{classG}
to the class of graph~$\Gamma_0$ by describing a polynomial-time, isomorphism-preserving
transformation that maps any connected graph~$G$ into a graph in~$\GG$.
Before doing so, we need to prove a small technical
lemma.

\begin{lemma}\label{lemma:div2div3}
  For each $n \in \Z$ there exists $d \in \{0, 1, 2, 3\}$ such that
  $(n+d)$ is not divisible by two nor by three; formally, $(n+d) \bmod
  2 \neq 0$ and $(n+d) \bmod 3 \neq 0$.
\end{lemma}
\begin{proof}
  Let us denote with $\llbracket n \rrbracket_k$ the equivalence class
  of all integers that are congruent to~$n$ modulo~$k$ (also called
  residual class): $\llbracket n \rrbracket_k = \{ \ldots, n-2k, n-k,
  n, n+k, n+2k, \ldots \}$.  The statement of the lemma can then be
  rephrased as: for each $n \in \Z$ there exists $d \in \{0, 1, 2,
  3\}$ such that $(n+d) \notin \left( \llbracket 0 \rrbracket_2 \cup \llbracket 0 \rrbracket_3 \right)$.

  \noindent
  Using well-known properties of residual classes we can derive the
  following table, that shows the value of~$d$ for any possible
  combination of residual classes modulo~$2$ and modulo~$3$ that~$n$
  may belong to.
  
  \begin{center}
    \begin{tabular}{c|ccc}
      $n$ & $\llbracket 0 \rrbracket_3$ & $\llbracket 1 \rrbracket_3$
      & $\llbracket 2 \rrbracket_3$ \\ \cmidrule{1-4} $\llbracket 0
      \rrbracket_2$ & $d=1$ & $d=1$ & $d=3$ \\ $\llbracket 1
      \rrbracket_2$ & $d=2$ & $d=0$ & $d=0$ \\
    \end{tabular}
  \end{center}

  \noindent
  For example, if $n \in \llbracket 1 \rrbracket_2 \cap \llbracket 0
  \rrbracket_3$, then $(n+2) \in \llbracket 1 \rrbracket_2$ and $(n+2)
  \in \llbracket 2 \rrbracket_3$.
  
\end{proof}

\begin{teo}\label{th:trasformazione}
  There exists a mapping $f: \Gamma_0 \rightarrow \GG$ such
  that for every two connected graphs $G_1, G_2 \in \Gamma_0$ with the
  same number of nodes and edges, $G_1 \Isomorph G_2$ if and only if $f(G_1) \Isomorph
  f(G_2)$. Furthermore, $f(G)$ can be computed in polynomial time with
  respect to the size of~$G$.
\end{teo}

\begin{proof}
  Given a connected graph $G \in \Gamma_0$, let us define $G' = f(G)$.
  We show how $G'$ is computed. Let $m=|E(G)|$ and
  $n=|V(G)|$. According to the Bertrand-Chebyshev
  theorem~\cite{inbookAZ2010}, for any integer $q>1$ there exists a
  prime in the set $\{q+1, \ldots, 2q-1\}$, and such prime can be
  found in polynomial time~\cite{tao2012}. Therefore, there is a
  prime~$p$ such that $2n < p < 4n$.

  \noindent
  The vertex set of~$G'$ is defined as
  
  \begin{equation*}
    V(G') = V(G) \cup \{v_{n+1}, v_{n+2}, \ldots, v_{p}\} 
  \end{equation*}

  \noindent where $v_{n+1}, v_{n+2}, \ldots v_{p}$ are new nodes.
  
  \usetikzlibrary{calc}
\newcommand\irregularcircle[2]{
  let \n1 = {(#1)+rand*(#2)} in
  +(0:\n1)
  \foreach \a in {10,20,...,350}{
    let \n1 = {(#1)+rand*(#2)} in
    -- +(\a:\n1)
  } -- cycle
}

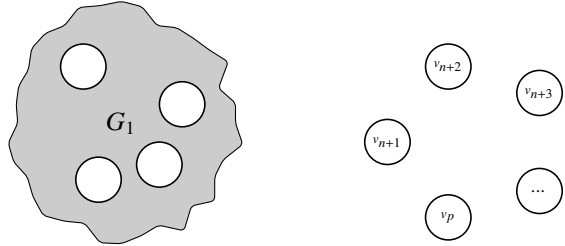
\begin{center}
\begin{tikzpicture}
\tikzstyle{vertex}=[auto=left,circle,draw=black,semithick,minimum size=17pt,inner sep=2pt]
\tikzstyle{vertexw}=[auto=left,circle,draw=black,fill=white,semithick,minimum size=17pt,inner sep=2pt]

  \draw[black,fill=black!20,rounded corners=.5mm] (-1.5, 1.25) \irregularcircle{1.5cm}{1.5mm};
  \node (g1) at (-1.5, 1.25) {$G_1$};
  \node[vertexw] (ga) at (-2, 2) {};
  \node[vertexw] (gb) at (-1.8, 0.5) {};
  \node[vertexw] (gc) at (-1, 0.7) {};
  \node[vertexw] (gd) at (-0.7, 1.5) {};


  \node[vertex] (n1) at (2, 1) {$\scriptscriptstyle v_{n+1}$};
  \node[vertex] (n2) at (2.8,2)  {$\scriptscriptstyle v_{n+2}$};
  \node[vertex] (n3) at (4,1.65)  {$\scriptscriptstyle v_{n+3}$};
  \node[vertex] (n4) at (4, 0.35) {$\scriptstyle \dots$};
  \node[vertex] (n5) at (2.8, 0) {{$\scriptscriptstyle v_{p}$}};

\end{tikzpicture}
\end{center}

  \noindent
  The edge set~$E(G')$ is constructed incrementally from~$E(G)$, as
  follows. Let
  \begin{align*}
    C_f &= \left\lbrace \{x, v_{n+1}\}\ |\ x \in V(G) \right\rbrace \\
    C_c &= \left\lbrace \{v_{n+1}, v_{n+2}\}, \{v_{n+2}, v_{n+3}\}, \ldots, \{v_p, v_{n+1} \right\rbrace
  \end{align*}

  \noindent that is, $C_f$ is a set of new edges that connect each
  node in~$V(G)$ to the first newly created node~$v_{n+1}$, and~$C_c$
  is a set of new edges that form a cycle within the new nodes. Since
  we have chosen~$p$ such that $2n < p < 4n$, the length of the cycle
  in~$C_f$ is greater than~$n$.

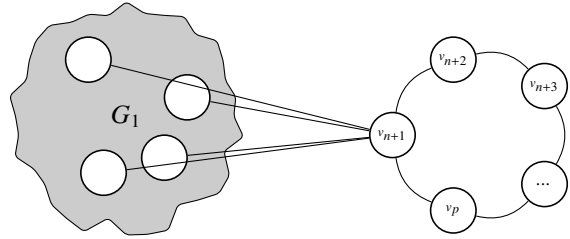
  \begin{center}
\begin{tikzpicture}
\tikzstyle{vertex}=[auto=left,circle,draw=black,semithick,minimum size=17pt,inner sep=2pt]
\tikzstyle{vertexw}=[auto=left,circle,draw=black,fill=white,semithick,minimum size=17pt,inner sep=2pt]

  \draw[black,fill=black!20,rounded corners=.5mm] (-1.5, 1.25) \irregularcircle{1.5cm}{1.5mm};
  \node (g1) at (-1.5, 1.25) {$G_1$};
  \node[vertexw] (ga) at (-2, 2) {};
  \node[vertexw] (gb) at (-1.8, 0.5) {};
  \node[vertexw] (gc) at (-1, 0.7) {};
  \node[vertexw] (gd) at (-0.7, 1.5) {};


  \node[vertex] (n1) at (2, 1) {$\scriptscriptstyle v_{n+1}$};
  \node[vertex] (n2) at (2.8,2)  {$\scriptscriptstyle v_{n+2}$};
  \node[vertex] (n3) at (4,1.65)  {$\scriptscriptstyle v_{n+3}$};
  \node[vertex] (n4) at (4, 0.35) {$\scriptstyle \dots$};
  \node[vertex] (n5) at (2.8, 0) {{$\scriptscriptstyle v_{p}$}};

  \path (n1) edge[bend left] node {} (n2);
  \path (n2) edge[bend left] node {} (n3);
  \path (n3) edge[bend left] node {} (n4);
  \path (n4) edge[bend left] node {} (n5);
  \path (n5) edge[bend left] node {} (n1);

  \foreach \from/\to in {n1/ga,n1/gb,n1/gc,n1/gd}
  \draw (\from) -- (\to);

\end{tikzpicture}
\end{center}

  \noindent
  We finally add a number~$s$ of self-loops within the new nodes in order to
  meet conditions~\ref{prop:div2} and~\ref{prop:div3}. Lemma~\ref{lemma:div2div3} guarantees that~$s$ is at most~$3$. Thus:
  \begin{align*}
    s=0 \implies C_s &= \emptyset\\
    s=1 \implies C_s &= \left\lbrace \{v_{n+2}, v_{n+2}\} \right\rbrace\\
    s=2 \implies C_s &= \left\lbrace \{v_{n+2}, v_{n+2}\},\{v_{n+3}, v_{n+3}\} \right\rbrace\\
    s=3 \implies C_s &= \left\lbrace \{v_{n+2}, v_{n+2}\},\{v_{n+3}, v_{n+3}\},\{v_{n+4}, v_{n+4}\} \right\rbrace
  \end{align*}

  \noindent
  The edge set~$E(G')$ is therefore
  defined as $E(G') = E(G) \cup C_f \cup C_c \cup C_s$.

  \noindent
  Observe that, since~$G$ is connected, the edge subset~$C_f$ induces at
  least one odd cycle (specifically, at least one cycle of length
  three), and therefore $G'$ is not bipartite\footnote{A well known result in graph theory states that a graph is bipartite if and only if it has no odd cycles}. Therefore we conclude
  that $G' \in \GG$.

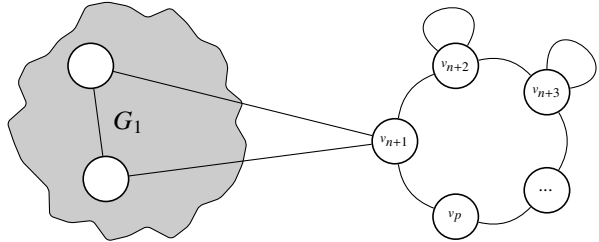
  \begin{center}
\begin{tikzpicture}
\tikzstyle{vertex}=[auto=left,circle,draw=black,semithick,minimum size=17pt,inner sep=2pt]
\tikzstyle{vertexw}=[auto=left,circle,draw=black,fill=white,semithick,minimum size=17pt,inner sep=2pt]


  \node (rnd1) at (-2, 2) {};
  \node (rnd2) at (-1.8, 0.5) {};
  
  \draw[black,fill=black!20,rounded corners=.5mm] (-1.5, 1.25) \irregularcircle{1.5cm}{1.5mm};
  \node (g1) at (-1.5, 1.25) {$G_1$};


  \node[vertex] (n1) at (2, 1) {$\scriptscriptstyle v_{n+1}$};
  \node[vertex] (n2) at (2.8,2)  {$\scriptscriptstyle v_{n+2}$};
  \node[vertex] (n3) at (4,1.65)  {$\scriptscriptstyle v_{n+3}$};
  \node[vertex] (n4) at (4, 0.35) {$\scriptstyle \dots$};
  \node[vertex] (n5) at (2.8, 0) {{$\scriptscriptstyle v_{p}$}};

  \path (n1) edge[bend left] node {} (n2);
  \path (n2) edge[bend left] node {} (n3);
  \path (n3) edge[bend left] node {} (n4);
  \path (n4) edge[bend left] node {} (n5);
  \path (n5) edge[bend left] node {} (n1);
  \path (n2) edge [out=140,in=60,looseness=6] node[above] {} (n2);
  \path (n3) edge [out=90,in=10,looseness=6] node[above] {} (n3);

  \foreach \from/\to in {n1/rnd1,n1/rnd2}
  \draw (\from) -- (\to);
  
  \node[vertexw] (ga) at (-2, 2) {};
  \node[vertexw] (gb) at (-1.8, 0.5) {};
  \draw (ga) -- (gb);

	
\end{tikzpicture}
\end{center}
  
  \noindent
  We now prove that $G_1 \Isomorph G_2 \iff f(G_1) \Isomorph f(G_2)$.
  
  \noindent ($\implies$) Assume $G_1 \Isomorph G_2$. Then, the
  isomorphism can be trivially extended to~$f(G_1)$ and~$f(G_2)$ since
  these graphs are obtained from $G_1, G_2$ by adding an identical
  structure.

  \noindent ($\Longleftarrow$) Assume $f(G_1) \Isomorph f(G_2)$. The
  only possible isomorphisms are those that map one of the 
  cycles~$C_c$ to the corresponding one on the other graph.
  Since in the transformation we have chosen $p > 2n$, the cycle
  length of~$C_c$ is larger than~$n$, and therefore is larger than any
  simple cycle in~$G_1$ (or~$G_2$). Consequently, the isomorphism
  between $f(G_1)$ and $f(G_2)$ can be restricted to an isomorphism
  between~$G_1$ and~$G_2$.
\end{proof}

Theorem~\ref{th:trasformazione} allows us to assert the main result of
this paper, that is the relation between primality test and graphs
isomorphism.

In what follows we denote by $\conn$ and $\unc$ be the sets of connected and unconnected graphs, respectively 
and by $\nonbip$ be the set of nonbiparite graphs.

\begin{teo}\label{th:GIhard}
  \begin{equation*}
    \text{\gi{$C$}} \mreduce \text{\compositeness{$\unc \cap \nonbip$}}
  \end{equation*}
\end{teo}

\begin{proof}
  Assume that there exists an algorithm ${\cal{A}}$ that solves the
  \compositeness{$\unc \cap \nonbip$} 
 decision problem.
 Then, the following algorithm solves the
  \gi{$\conn$} decision problem 
  and, at the same time,
  provides a polynomial time many-one reduction 
  from $\gi{\conn}$ to \compositeness{$\unc \cap \nonbip$}.

\begin{codebox}
\Procname{$\proc{Graph-Isomorphism}(G_1,G_2)$}     
 \li 	\If $|V(G_1)| \neq |V(G_2)|$ {\bf or} $|E(G_1)| \neq |E(G_2)|$
 \li 	\Then \Return \textsc{no} \End
 \li    $G_3 \gets f(G_1)$ \Comment Theorem \ref{th:trasformazione}
 \li    $G_4 \gets f(G_2)$ \Comment Theorem \ref{th:trasformazione}
 \li 	$G \gets G_3 \cup G_4$ \Comment graphs disjoint union
 \li	\Return ${\cal{A}} (G)$
 \end{codebox}

\noindent
In fact,
by Theorem~\ref{th:trasformazione},
$G_1$ is isomorphic to~$G_2$ if and only if~$f(G_1)$ is isomorphic to~$f(G_2)$.
Since both~$f(G_1)$ and~$f(G_2)$ belong to $\GG$, then by Theorem~\ref{classG},
 $G = G_3 \cup G_4$ is decomposable
if and only if~$G_3$ is isomorphic to~$G_4$.

\noindent
It is easy to verify that $\proc{Graph-Isomorphism}$ is a polynomial time many-one reduction.

\end{proof}

Note that
\compositeness{$\unc \cap \nonbip$} remains $GI$-hard even if we relax the \emph{undirected}
constraint or the \emph{nonbipartite} one, as the resulting class of graphs would be larger than the one which was considered throughout our discussion.

\begin{corol}\label{corol:turing-reduction}
\primality{$\unc \cap \nonbip$} is $GI$-hard or, equvalently,
  \begin{equation*}
    \text{\gi{$\conn$}} \treduce \text{\primality{$\unc \cap \nonbip$}}
  \end{equation*}
\end{corol}

\begin{proof}
The proof follows directly from the proof of Theorem~\ref{th:GIhard} by inverting the result
provided by the  \emph{oracle} ${\cal{A}}$.
\end{proof}
\section{Conclusions}\label{sec:conclusions}
In this paper we proved that primality testing of unconnected, nonbipartite  grahps with respect 
to direct product is at least as hard as deciding graph isomorphism. The same result also applies
to the computation of a prime factorization of a graph.
This result answer a long standing open question posed in~\cite{imrich1998} and shows
the crucial role played by connectedness in decomposing a graph.

It would be of some interest to investigate the reversed question, i.e.,  
whether deciding graph isomorphism is at least as hard as primality testing or not.
Another interesting  research direction is the study and the implementation
of efficient heuristics for computing a prime factorization or its approximation
of large, possibly unconnected and/or weighted graphs knowing that a polynomial time algorithm for computing such a prime factorization
is unlikely to exist.




\bibliographystyle{elsarticle-num}
\bibliography{manuscript}

\end{document}